\documentclass[a4paper]{article}

\usepackage[utf8]{inputenc}
\usepackage[colorlinks]{hyperref}
\usepackage{csquotes}
\usepackage[ruled,vlined]{algorithm2e}
\usepackage{tikz}
\usepackage[margin=22mm]{geometry}
\usepackage{graphicx}
\def\shrug{\texttt{\raisebox{0.75em}{\char`\_}\char`\\\char`\_\kern-0.5ex(\kern-0.25ex\raisebox{0.25ex}{\rotatebox{45}{\raisebox{-.75ex}"\kern-1.5ex\rotatebox{-90})}}\kern-0.5ex)\kern-0.5ex\char`\_/\raisebox{0.75em}{\char`\_}}}

\usepackage{lmodern}
\usepackage{array}
\usepackage{graphicx}
\graphicspath{{./pic/}}
\usepackage{amssymb, amsfonts, amsmath, amsthm, mathrsfs}
\usepackage{enumerate}
\usepackage{enumitem}
\usepackage{epstopdf}
\usepackage[font=footnotesize]{caption}
\usepackage[colorlinks]{hyperref}
\usepackage{tikz}

\theoremstyle{plain}
\newtheorem{theorem}{Theorem}[section]

\newtheorem{lemma}[theorem]{Lemma}
\newtheorem{proposition}[theorem]{Proposition}

\theoremstyle{definition}
\newtheorem{remark}[theorem]{Remark}

\numberwithin{equation}{section}

\newcommand{\N}{\mathbb{N}}

\newcommand{\ind}[1]{\mathbf{1}_{\left\{#1\right\}}}

\newcommand{\ceil}[1]{{\left\lceil #1 \right\rceil}}

\renewcommand{\bar}[1]{\overline{#1}}

\renewcommand{\phi}{\varphi}
\renewcommand{\epsilon}{\varepsilon}

\newcommand{\E}{\mathbf{E}}

\renewcommand{\P}{\mathbf{P}}

\newcommand{\dd}{\mathrm{d}}

\renewcommand{\rho}{\varrho}
\renewcommand{\epsilon}{\varepsilon}

\usepackage{xcolor}
\usepackage[normalem]{ulem}
\usepackage{soul}

\title{A tractable non-adaptative group testing method for non-binary measurements}

\author{Emilien Joly\thanks{CIMAT, Guanajuato, Mexico, \texttt{emilien.joly@cimat.mx}.} \and Bastien Mallein\thanks{Université Sorbonne Paris Nord, LAGA, UMR 7539, F-93430, Villetaneuse, France. Member of the MODCOV19 plateform and the GROUPOOL initiative, \texttt{mallein@math.univ-paris13.fr}.}}

\date{\today}

\begin{document}

\maketitle

\begin{abstract}
The original problem of group testing consists in the identification of defective items in a collection, by applying tests on groups of items that detect the presence of at least one defective item in the group. The aim is then to identify all defective items of the collection with as few tests as possible. This problem is relevant in several fields, among which biology and computer sciences. In the present article we consider that the tests applied to groups of items returns a \emph{load}, measuring how defective the most defective item of the group is. In this setting, we propose a simple non-adaptative algorithm allowing the detection of all defective items of the collection. This method improves on classical group testing algorithms using only the binary response of the test.

Group testing recently gained attraction as a potential tool to solve a shortage of COVID-19 test kits, in particular for RT-qPCR. These tests return the viral load of the sample and the viral load varies greatly among individuals. Therefore our model presents some of the key features of this problem. We aim at using the extra piece of information that represents the viral load to construct a one-stage pool testing algorithm on this idealized version. We show that under the right conditions, the total number of tests needed to detect contaminated samples can be drastically diminished.
\end{abstract}

\section{Introduction}

The group testing problem consists in identifying a subset of defective items among a larger set by using tests on pools of items answering the question ``Does this pool contains at least one defective item?''. This problem has a long history, and appeared several times in different fields of medical biology \cite{Dor43,Tho62,ThM06,FHG12} and computer sciences \cite{MTT08,IKW18,AJS19}. It has also been the subject of an important mathematical literature, which studied optimal algorithms for the detection of defective items with minimal use of tests, which are considered a limiting resource. Those algorithms can be divided into two main categories:
\begin{description}
  \item[adaptive testing:] in which the choice of a pool is influenced by the previous results of tests applied to the group;
  \item[non-adaptive testing:] in which the choice of a pool does not depend on the results of previous tests.
\end{description}
The aim of a pool testing algorithm is to assess, as precisely as possible, the status (defective or not) of each item, through the tests made on pools of items, while using as few tests as possible.

It should be clear that the more permissive adaptive testing option allows for more flexibility, a more parsimonious use of tests can thus be archived in this setting. At one extreme, a search tree can be used to detect defective items with a maximal economy of tests \cite{CJBJ}. In contrast, non-adaptive testing allows the possibility to massively parallelize the procedure. As all pools can be constructed before any result is known, all tests can be performed simultaneously, which can decrease the time needed to obtain the result. Moreover, in the context of biological testing, non-adaptive schemes decrease the risk of contamination or of decay of samples during their treatment.

It might be noted that several types of adaptive testing allow some level of parallelizing. For example, two- or three-stages algorithms can be considered. In this situation, a first set of pools is constructed without prior information. Using the result of testing on these pools, a second set of pools is constructed. With the tests made on this second set of pools, the status of each item is assessed in a two-stage algorithm, or a third set of pools is constructed and tested in a third stage algorithm, before assessing the status of the items.

One of the first pool testing algorithms to be describe was introduced by Dorfman \cite{Dor43}, as a method to detect syphilis in recruited US soldiers. This algorithm is the following: samples taken from individuals are pooled together in a group, which is then tested for syphilis. If the pool turns negative, all individuals are declared non-contaminated, while if the pool turns positive, then each individual of the pool is tested. Note that this is a two-stage algorithm, which we refer to as Dorfman's algorithm.

Several adaptive and non-adaptive pool testing algorithms have been described over the years, such as matrix testing \cite{CCK99}, smart testing \cite{ThM06}, and testing based on risk estimation for items \cite{ABB19,BBC20}. We refer to \cite{AJS19} for a recent survey on this topic. To compare these algorithms, it is necessary to specify more precisely the context in which they are used, such as the relative number of defective and non-defective items, the authorized false positive and false negative rates, etc.

\paragraph{Prevalence and efficiency of pooling procedures.}
As stated above, the objective of pool testing is the reduction of the number of tests used on a population of $N$ items in order to identify the defective ones. If an algorithm uses a total of $T$ tests, we measure its resource-based efficiency by the quantity
\begin{equation*}
  E=\frac{T}{N}.
\end{equation*}
This ratio measures the average number of tests used per item in this pool testing algorithm to detect the defective ones. Therefore, the lower this ratio is, the more parsimonious the algorithm.

Observe that any reasonable algorithm of pool testing should verify $E<1$, as otherwise the testing of any item separately represents a more efficient use of resources. In the present article, we assume that a known proportion $p$ of items is defective. It is worth noting that in that situation, a lower bound on the efficiency of a reasonable non-adaptive pool testing algorithm is $E(p) \geq p$. Indeed, there are approximately $pN$ defective items among $N$, so if one makes less than $pN$ tests, there is no possibility to detect the defective items if all pools contains at least one defective. One is interested in the optimal dependency of $E$ in the parameter $p$.

The optimal efficiency of the Dorfman algorithm previously described is obtained by choosing the size of the pool depending on the value of $p$ in such a way that it is minimal. It can be computed as follows
\begin{equation}
  \label{eqn:optDor}
  E^D(p) = \min_{n \in \N} \frac{1 + n ( 1 - (1-p)^n )}{n} \sim 2 p^{1/2} \quad \text{ as } p \to 0.
\end{equation}
Indeed, if one creates pools of $n$ items, one test is required \emph{for the pool} to detect if defective items are present or not, and if the pool is positive (which happens with probability $1 - (1-p)^n$), an additional one is needed \emph{per item}. The equivalent is obtained by choosing $n \approx p^{-1/2}$ as $p \to 0$.

Mézard et al. \cite{MTT08} constructed asymptotically optimal non-adaptive and two-steps pool testing algorithms, which detects asymptotically all defective items, while keeping an efficiency of
\begin{equation}
  \label{eqn:optMez}
  E^*(p) \sim C_* p |\log p| \text{ as } p \to 0,
\end{equation}
for some $C_* > 0$. This algorithm is based on the construction of random pools of size $n \approx c p^{-1}$ of items, such that each item belongs to $L \approx C |\log p|$ pools. An item is declared non-defective if it belongs to at least one pool tested negative, is declared defective if it belongs to at least one positive pool with all the other items being declared non-defective, and is declared ambiguous otherwise. In that situation, depending of the value of $c,C$, either with high probability each defective item will belong to at least one pool of non-defective items, and thus be identified as non-defective, or with high probability, the number of ambiguous items after the first stage is small enough that they can be tested individually.

\paragraph{Pooling in the context of the COVID-19 epidemics.}
In the context of COVID-19, pool testing has been massively proposed and implemented as a method to diminish the marginal cost of a test as well as to answer local shortages of test kits, see for example \cite{BBC20, Ghosh2020, Shani-Narkiss2020, Ben-Ami2020, Hogan2020, SinnottArmstrong2020, Lipsitch2020, GG, Mutesa2020, Torres2020, Lohse2020} among many others. The necessity of early detection of contaminated individuals has been underlined many times, in particular due to the large number of presymptomatic, asymptomatic and mildly symptomatic individuals that remain contagious and can carry the diseases to vulnerable people. As a result, the demand for effective and quick testing has skyrocketed, with the offer being limited by the number of test kits and trained medical professionals for the sampling. The question of optimization of pool testing thus has practical consequences, as improving on the efficiency of a testing algorithm can increase the number of individuals that can be tested with the same number of kits.

A typical test used for the detection of contamination to SARS-COV-2 is the RT-qPCR test (or PCR test for short), the \emph{reverse-transcriptase quantitative polymerase chain reaction}. This test allows the measurement of the number of RNA segments typical of the virus that are present in a given sample (usually, three different RNA segments are tested simultaneously to improve on the measure), which is related to the viral load of the sample. As the name suggest, the measure is quantitative, thus returns more than binary response (which would be akin to a defective/not defective result in the classical pooling literature). As such, it seems that this additional piece of information could be used to improve on the existing group testing strategies to reduce the the number of tests needed for detection.

However, let us underline a couple of important caveats. First, the quantity measured by the PCR test is related to the logarithm of the viral load carried in the sample, rather than the viral load itself, with some noise on the measure \cite{BMR}. Therefore the exact viral load is not known, but rather its order of magnitude. Secondly, the viral load in defective items spans a large range, of several orders of magnitude \cite{Jones2020}. Therefore, if two defective items with viral loads $c_1$ and $c_2$ are tested in the same pool, the result of the measure will be
\begin{equation}
  \label{eqn:maxStable}
  \log (c_1 + c_2) \approx \max(\log c_1, \log c_2),
\end{equation}
as $c_1$ and $c_2$ will typically be of different orders of magnitude.

The aim of this article is to propose and study an algorithm that uses the viral load of an item to improve its efficiency. We construct this algorithm on an idealized version of the situation described above. We discuss in more details in Section~\ref{sec:covid} the adaptation of the algorithm to the COVID scenario, pointing some of its limitations.

\paragraph{Defective items with load.}
We consider in this article some theoretical aspects of pooling strategies that can be employed for the detection of defective items with load, in order to adapt to the PCR testing scenario previously described. We assume here that each defective item $u$ has a positive value $x_u$ attached that we call its load. A non-defective item will have a load of $0$. The test of a pool $A$ of items has the effect of measuring the value $\max_{u \in A} x_u$, i.e. the largest load among all items in the set $A$.

Observe that if the load of items belongs to $\{0,1\}$, then we are in the settings of the classical pool testing, and a test only detects the presence of at least a contaminated item. However, if this load can takes more values, we show that the results of several tests can be crossed to extract additional information on the items. The load $x_u$ can be thought of as the logarithm of the viral load of an individual in PCR settings, and the choice of measuring the maximal load of a set comes from \eqref{eqn:maxStable}.

We denote by $p$ the \emph{prevalence} of defective items (i.e. the proportion of defective items in the set to be tested). We assume here that $p$ is known (or at least adequately estimated), so it can be used to choose the size and number of pools to be made\footnote{Even if the true proportion of contaminated samples is not known, an upper bound for our method to work successfully even if not at full capacity. Note that we only take into account here the proportion of defective samples currently being tested and not the more complicate question of estimating the time-dependent prevalence of defectiveness in the population.}. The load associated to each item can then be written as $x_u = \xi_u Z_u$, where $\xi_u$ is a Bernoulli random variable with parameter $p$ representing the fact that item $u$ is defective or not, and $Z_u$ is an independent $[0,1]$-valued random variable. In this article we will consider $Z_u$ uniformly distributed either on $[0,1]$ or on $\{1/K,2/K,\ldots, 1\}$ for some $K \in \N$. Note that we assume that all defective items have a positive load, but in real-world examples, there are limitations on the accuracy of the detection and some defective items would have a load of $0$. We do not try to measure these false negative as they are present no matter the testing method used.\footnote{To take into account false positives due to the limits of the detection method, one could choose instead to consider $Z_u$ uniform on $\{0,1/K,2/K,\ldots, 1\}$, with $\xi_u=1$ and $Z_u=0$ corresponding to undetectable defectives.}

The quantity $K$ described above can be interpreted as the level of precision of the measure. The larger $K$ is, the easier it is to distinguish the level of two defective items with similar loads. As a result, the efficiency attained by our algorithm will decay as $K$ increases, to attain optimal efficiency when $K = \infty$, which corresponds to $Z_u$ uniformly distributed on $[0,1]$.

The model of group testing with load allow us to explore the information carried by the maximal value of a set of items. It interpolates with the classical group testing model when $K = 1$, and its limit as $K \to \infty$ is a universal problem, in the sense that all atomless distribution for $x_u$ would create the same combinatorial problem. The case \(K>2\) that generalizes the zero-one binary information corresponding to the healthy-defective alternative is already present in \cite{EM16} where the load of a pool is supposed to take the form of the sum of the load of each individual. The closest case to our study is the one of \cite[Chap 11.3]{DH00} and \cite{HX87} for \(K=2\) with the limit that only one defective (of load 1) and one mediocre (of load $1/2$) are allowed in the sample. In essence, the linear case considered in \cite{EM16} and some generalizations described in \cite[p119]{AJS19} are simpler to study than the multilevel loads combined with the (non-linear) maximum load in the spirit of \eqref{eqn:maxStable}. Indeed, a lot of the information is lost in only considering the maximum so that the small loads are more likely to be hard to detect. See the results in Section \ref{sec:computationFNR-FPR} for precise explanations of this fact.

\paragraph{Organization of the paper.}
We propose here a very simple one-step (non-adaptive) algorithm for the detection of defectives, sometimes under the assumption that the same sample cannot be placed in more than $L$ pools. This algorithm is asymptotically efficient as $p \to 0$ while remaining simple to implement and to evaluate. We describe in the next section the general form of the algorithm we study. We then show how to optimize this algorithm assuming that each sample can only be part of a finite number of pools in Section~\ref{sec:petitL}, and optimal efficiency that can be obtained by this algorithm in Section~\ref{sec:Lmax}. We then provide some numerical simulations to compare these asymptotic results to their finite value counterpart in Section~\ref{sec:simulations}.

\section{The Grid Pool Testing algorithm}

In this work, we focus on a simple one-step non-adaptive algorithm. In this algorithm, items are organized on a grid, and the pools are made of the lines, columns and the diagonals of different slopes of this grid. The algorithm mainly focus on reconstructing the status of items from the measures made on these diagonals. The parameters of the algorithms to optimize are the size of the grid (representing the number of items in each pool) and the number of diagonals slopes to consider (representing the number of pools each item belongs to).

\paragraph{Defining the grid.}
Before describing the algorithm in more details, we introduce some notation. We assume the number of items to test to be sufficiently large that it is possible to divide them into batches of $n^2$ items. We describe the algorithm on a given batch.

The items are dispatched on a grid $n \times n$, with each item being identified by its position $(i,j) \in \{1, \ldots, n\}^2$. We write $\xi_{i,j} = 1$ if $(i,j)$ is defective and $\xi_{i,j} = 0$ otherwise. Moreover we denote by $X_{i,j}$ the load of the item (which is $0$ if the item is non-defective, or a number in $(0,1]$ otherwise).
With the modelling of the previous section, we note that $(\xi_{i,j}, 1 \leq i, j \leq n)$ are i.i.d. $\mathcal{B}(p)$ random variables, with $p$ the proportion of defective. Conditionally on $\xi$, $(X_{i,j}, 1 \leq i,j \leq n)$ are independent random variables, with $X_{i,j}=0$ if $\xi_{i,j} = 0$ and $X_{i,j}$ uniformly distributed on $(0,1]$ or on $\{1/K,2/K\ldots, 1\}$, depending on the context.

\paragraph{Defining the pools.}
The pools used can loosely be described as the diagonals of the grid. More precisely, we introduce the following sets of $n$ items to construct the pools of the algorithm:
\begin{itemize}
  \item the lines $L_i = \{(i,k), 1 \leq k \leq n\}$, for $1\leq i \leq n$.
  \item the columns $C_j = \{(k,j), 1 \leq k \leq n\}$, for $1\leq j \leq n$.
  \item the diagonals with various slopes $D^{a}_b = \{(k,ak+b \text{ mod}(n)), 1 \leq k \leq n\}$ for $1\leq b \leq n$, where $a\in \{1,\dots,n-1\}$.
\end{itemize}
In an algorithm constructed such that each item is part of $L$ pools, the pools will be taken as families of lines, columns and diagonals with slopes smaller than $L-2$.
In the rest of the article we will assume this family of pools will form a $N(n^2,n,L)$ multipool, in the terminology of \cite{Tau20}. In other words, we need our pools to satisfy the following three properties:
\begin{enumerate}
  \item each pool contains exactly $n$ items;
  \item each item belongs to exactly $L$ pools;
  \item two items $(i,j)$ and $(k,l)$ share at most one pool in common.
\end{enumerate}
While the first two properties are straightforward from the definition, the third one is not, and only holds under some assumptions on $n$ and $L$.

\begin{lemma}
\label{lem:nombre}
The family $\{L_k, C_k, D^a_k, 1 \leq k \leq n, a \leq L-2\}$ is a $N(n^2,n,L)$ multipool if and only if $L-2$ is smaller than the smallest prime divisor of $n$.
\end{lemma}

\begin{proof}
We first note that two line never cross, and that a line crosses with a column or a diagonal at exactly one point. Therefore, to verify that $\{L_k, C_k, D^a_k, 1 \leq k \leq n, a \leq L-2\}$ is a multipool, it is enough to check that no too diagonal cross at more than one place (treating columns as diagonals of line $0$).

Observe that for $a \neq b$, two diagonals $D^a_k$ and $D^b_\ell$ cross at a point $(i,j)$ such that $k + a i \equiv \ell + b i \text{ mod } n$, i.e. such that $(b-a) i \equiv \ell - k \text{ mod } n$. By the fundamental theorem of algebra, there exists a unique $i \in [1,n]$ satisfying this property if and only if $(b-a)$ is prime with $n$. As $|b-a| \leq L-2$ is smaller than the smallest prime factor of $n$, we deduce this is indeed the case, proving that any two pools cross at either $0$ (if they have the same slope) or $1$ point.
\end{proof}

\begin{remark}
More generally we could prove that selecting families of lines, columns and diagonals in the $n \times n$ grid, it is possible to create a $N(n^2,n,L)$ multipool if and only if $L-2$ is smaller than the smallest prime divisor of $n$.
\end{remark}

In the rest of the article, we enumerate the pools as the family $\{P_j, j \leq nL\}$, with $P_1,\ldots P_n$ corresponding to the lines, $P_{n+1}, \ldots P_{2n}$ to the columns and the rest to the diagonals, in the increasing order of their slope. For each $\ell \leq nL$, the effect of probing the pool $P_\ell$ corresponds to the action of discovering the value $V_\ell := \max_{(i,j) \in P_\ell} X_{i,j}$, the largest load among all defective items belonging to the pool. Finally, for convenience, we denote \(\mathcal{P}_{i,j}\) the set of pools associated to the item \((i,j)\),
\begin{equation*}
  \mathcal{P}_{i,j} = \{\ell : (i,j)\in P_\ell\}.
\end{equation*}

\paragraph{Computation of the positives.}
The final step of the algorithm consists in a reconstruction of the load of each item via the information contained in the family $\{V_\ell, \ell \leq nL \}$. We observe immediately that if $V_\ell = 0$, then all items in the pool are non-defective, and if $V_\ell = x \neq 0$, then there exists at least one item in the pool with load equal to $x$.

To reconstruct the load of each item, we employ the following procedure.
\begin{enumerate}
  \item For every item \((i,j)\), let \(V_{i,j}=\min_{\ell: (i,j)\in P_{\ell}} V_{\ell}\).
  \item If \(V_{i,j} =0\), the item \((i,j)\) is declared negative.
  \item Otherwise, we count the number of apparitions of the value \(V_{i,j}\) inside of the pools containing \((i,j)\) : \(I_{i,j}=|\{\ell: (i,j)\in P_{\ell}\text{ and }V_{\ell}=V_{i,j}\}|\).
  \begin{enumerate}
    \item If \(I_{i,j}\geq 2\), meaning that at least two tests containing item \((i,j)\) measured it with the same value, the item \((i,j)\) is declared positive.
    \item Otherwise, the item \((i,j)\) is declared negative.
  \end{enumerate}
\end{enumerate}

Here is the reason behind this definition. By the assumptions we made on the test, $V_{i,j}$ is an upper bound for the load $X_{i,j}$ of the item. In particular, if $V_{i,j} = 0$, we label the item as non-defective. However, if $V_{i,j} > 0$ it might be that the item has been, by chance, mixed with defective items in all the tests that were made on it. The fact that level $V_{i,j}$ is attained at least twice is a much stronger indication of the defectiveness of $(i,j)$, as a false positive in that case would mean that it has been by chance mixed in two pools with different defective items sharing exactly the same load, and that in all other pools, there was at least one item with a larger load. In the asymptotic we will consider, this will not occur with large probability, and similarly if $I_{i,j} = 1$, with high probability the item will be negative.

\begin{remark}
Observe the procedure we describe here to assess the load and status of each item is not the most accurate. With extra care, one could gain more precision of the reconstruction, for example by checking that each measured load in the pools has been associated to at least one item. However, the procedure described here has the advantage of simplicity and locality: to give the status of an item, one has only to consider the results of the tests related to this item. This makes the forthcoming computation of the probability that an item is wrongfully characterized significantly easier, and it remains efficient enough in the range of parameters we consider.
\end{remark}

We sum up the complete procedure inside Algorithm~\ref{algo:main} and a concrete toy example in Figure~\ref{fig:example}.
\begin{algorithm}[h]
  \caption{Grid Pool Testing}
  \label{algo:main}
\SetAlgoLined
 \textbf{Parameters:} n,L,\\
 \textbf{Inputs:}  \({\bf X}=(X_1,\dots,X_{n^2})\)\\
 Store \({\bf X}\) inside the grid \((X_{i,j})_{1\le i,j\le n}\) line by line\;
 Define \(P_1,\dots,P_n\) as the lines, \(P_{n+1,\dots,P_{2n}}\) as the columns and \(P_{2n+1},\dots,P_{nL}\) as the diagonals\;
 Initialize a matrix \({\bf S}=(S_{i,j})_{i,j}\) of empty lists\;
 \For{\(\ell=1,\dots,nL\)}{
 Compute \(V_{\ell}=\max_{(i,j)\in P_{\ell}} X_{i,j}\)\;
 Append \(V_{\ell}\) to every \(S_{i,j}\) with \((i,j)\in P_{\ell}\)\;
 }
 Initialize a matrix \({\bf R}=(R_{i,j})_{i,j}\) of zeros\;
 \For{\(1\le i,j\le n\)}{
 Compute \(V_{i,j}=\min_{s\in S_{i,j}}s\)\;
 Compute \(I_{i,j}=\sum_{s\in S_{i,j}} \ind{s=V_{i,j}}\)\;
 \If{\(V_{i,j}\neq 0\) and \(I_{i,j}\geq 2\)}{
 Set \(R_{i,j}=1\)\;
 }
 }
 Store the matrix \({\bf R}\) line by line into a vector \((R_1,\dots,R_{n^2})\)\;
 \textbf{Result:} \((R_1,\dots,R_{n^2})\)
\end{algorithm}

\paragraph{Efficiency and optimization.}
It is worth noting that the algorithmic complexity of Algorithm~\ref{algo:main} is $O(n^2L)$. In terms of test usage, it is easy to compute the efficiency of this algorithm as there are a total of \(nL\) pools of \(n\) items that are tested, in an effort to detect defective elements among \(n^2\) items. The corresponding efficiency is then
\begin{equation*}
  E=\frac{nL}{n^2}=\frac{L}{n}.
\end{equation*}
To complete the study of this algorithm, one then need to compute its false positives (when \(R_{i,j}=1\) implying detection as defective while \(X_{i,j}=0\) so the item is non-defective) and false negatives (when \(R_{i,j} = 0\) whereas \(X_{i,j}>0\)) rates.

We denote by \(FPR\) (respectively \(FNR\)) the expected number of false positives and false negatives returned by this algorithm, divided by the total number of contaminated items. These two quantities depend on the four parameters \(p,K,n\) and \(L\) in an intricate fashion. However, note that while \(n\) and \(L\) are integer parameter of the algorithms we can choose, \(p \in [0,1]\) and \(K \in \bar{\N}\) are modelling parameters of the problem, representing respectively the proportion of defective items and the accuracy of the test. Therefore the main goal of this study is to optimize the efficiency $E$ of this algorithm by choosing the optimal \(n(p,K)\) and \(L(p,K)\) in a way that insures that \(FPR\) and \(FNR\) both stay below fixed quantities $\epsilon$ and $\delta$.
Our main results are considered under the asymptotic $p \to 0$ of a small proportion of defective items. But as most of the computations made are explicit before taking limits, computing the optimal value of \(n\) and \(L\) for given values \(p,K\) remains straightforward.

\begin{remark}
\label{rem:falseNegative}
Note that an item is falsely labelled as negative if it is part of at most one pool in which it is the item with the largest load. It corresponds to items at position $(i,j)$ such that $I_{i,j} = 1$ in the above algorithm. Therefore, items such that $I_{i,j}$ could be labelled as inconclusive and tested again in a separate batch in a two-steps algorithm with no false negative.
\end{remark}

\begin{figure}[h]
  \begin{center}
    \includegraphics[scale=0.7]{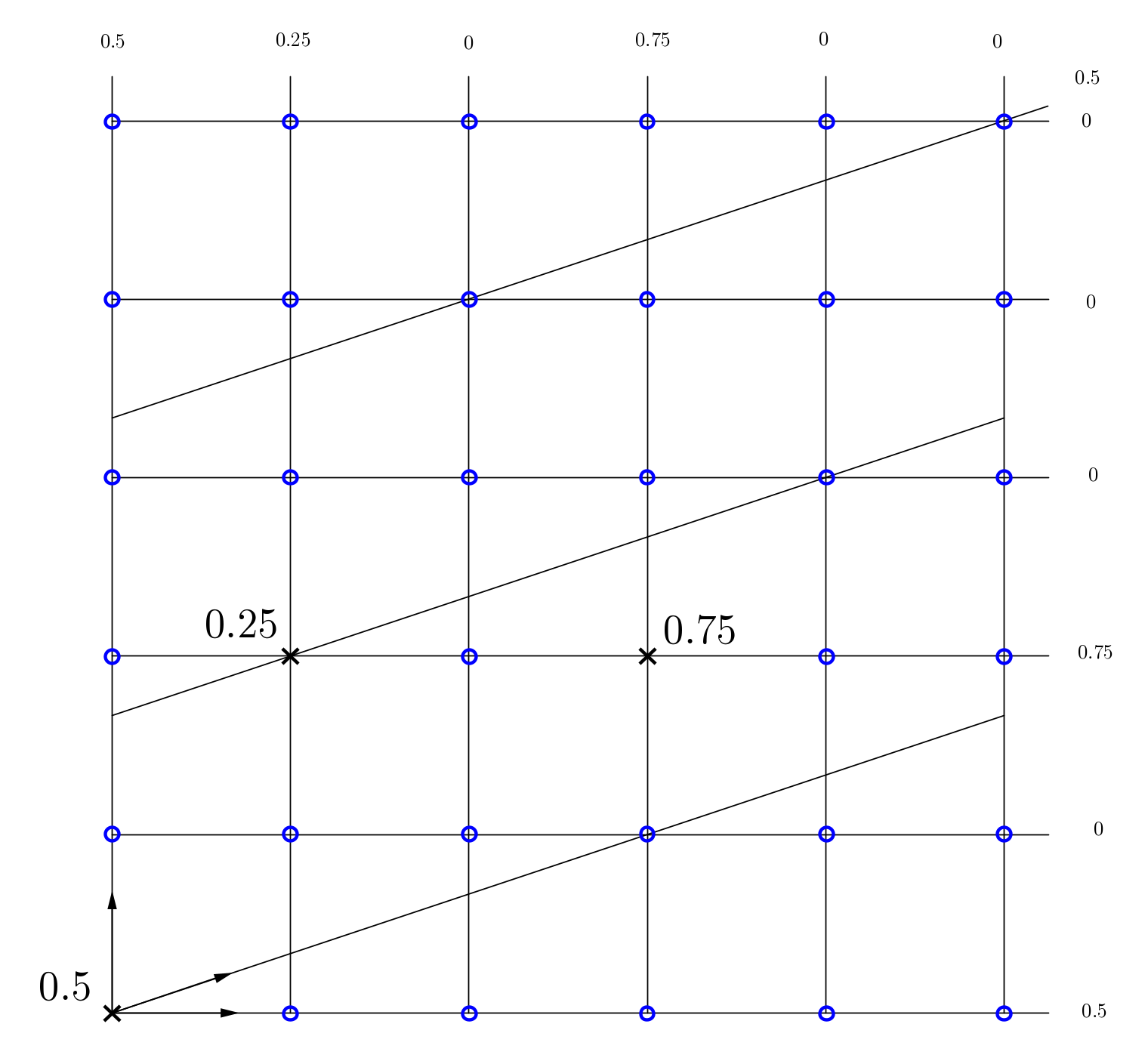}
    \caption{A grid with \(L=3\), \(n=6\), \(K=4\) and \(N=36\). There are 3 defective items of respective levels 0.25, 0.5 and 0.75. For the sake of clarity,  we only showed the test involving the bottom leftmost item for the diagonal of slope \(1/3\) (hence five more test are not represented here). The blue circle represent the healthy items whereas the black crosses represent the defective items for whom the level of defectiveness is specified.}
    \label{fig:example}
  \end{center}
\end{figure}

\section{Computation of the false negative and false positive rates}
\label{sec:computationFNR-FPR}
\subsection{Exact calculations}
In this section, we give explicit upper bounds on the false negative and false positive rates of the algorithm. This does not take into account a possible defective measurement of the loads of the pools.
The false negative rate of the algorithm is expressed as the probability that a contaminated item is not detected at the end of the algorithm whereas the false positive rate is the probability that a non-defective item is detected as defective. In Algorithm~\ref{algo:main}, it is straightforward to compute the false negative/positive rates, as the reconstructed status of an item only depend on the status of items sharing a pool with it. This algorithm being unchanged by changing the coordinates of the grid, as on a torus, these rates do not depend on the position $(i,j)$ of the item in the grid.
We thus only compute the false negative probability of item $(1,1)$, given its load. For all $x \in (0,1]$, we set
\begin{equation*}
  FN(x)=\P ((1,1)\text{ is declared negative }|X_{1,1}=x),
\end{equation*}
and
\begin{equation*}
  FP(x) = \P((1,1) \text{ is declared positive with load }x | X_{1,1} = 0).
\end{equation*}
The false negative rate \(FNR\) and the false positive rate \(FPR\) are then given by
\begin{align*}
  FNR(n,L;p,K) &=p\E\left( FN(\ceil{KU}/K) \right)\\
  FPR(n,L;p,K) &= (1-p)\sum_{k=1}^K FP(k/K)
\end{align*}
with $U$ a uniform random variable on $[0,1]$ and the convention that $\ceil{\infty U}/\infty = U$ and $FPR(n,L;p,\infty)=0$. We used here that a given item is defective with probability $p$ and non-defective with probability $1-p$. The following Proposition gives two upper bounds on the false positive/negative rates.

\begin{proposition}
  \label{prop:exact_upper_bounds}
  Let $x = k/K$ and let \(g_{n,p}(x)=(1-p(1-x))^{n-1}\). Then it holds that
  \begin{equation}
    \label{eqn:genericFalseNegative}
    FN(x) \leq  L\left( 1 - g_{n,p}(x) \right)^{L-1}
  \end{equation}
  and that
  \begin{equation}
    \label{eqn:genericFalsePositive}
    FP(x) \leq \frac{L(L-1)}{2} \left(\frac{np}{K}\right)^2 g_{n,p}(x)^2 \left( 1 - g_{n,p}(x) \right)^{L-2}.
  \end{equation}
  In particular, when \(K \to \infty\), the false positive rate \(FPR(n,L;p,K)\) tends to 0.
\end{proposition}

\begin{proof}
  We observe that the probability to wrongfully declare an item as negative depends on $K$ (in the discrete case) only through the fact that $x$ takes its values in $\{1/K,2/K\ldots, 1\}$. This allows us to give a unified expression for the upper bound of $FN$.

  Given $x$ the load of the item $(1,1)$, we note this item will be wrongfully declared negative in Algorithm~\ref{algo:main} if and only if $I_{1,1}=1$ (as $V_{1,1} \geq x > 0$ a.s.). Decomposing according to the test in $\mathcal{P}_{1,1}$ measuring the lowest viral load, we have
  \begin{equation*}
      FN(x) = \P(I_{1,1}=1|X_{1,1} = x) = \sum_{\ell \in \mathcal{P}_{1,1}} \P(V_\ell < \min_{{\ell' \neq \ell}} V_{\ell'} | X_{1,1}=x)
      = L \E(\phi(V_\ell)|X_{1,1}=x),
  \end{equation*}
  where $\phi(y) = \P(\min_{{\ell' \neq \ell}} V_{\ell'} > y| X_{1,1}=x)$ and $\ell_0$ is a fixed element of $\mathcal{P}_{1,1}$. Using that the $(V_\ell, \ell \in \mathcal{P}_{1,1})$ remain i.i.d. conditionally on $X_{1,1}= x$, thanks to Lemma~\ref{lem:nombre}, we have $\phi(y) = \P(V_{\ell_0} > y|X_{1,1} = x)^{L-1}$. Then, using that $\phi(y) \leq \phi(x)$ for all $y \geq x$, we obtain
  \begin{equation*}
    FN(x) \leq L \P(V_{\ell_0} > x|X_{1,1} = x)^{L-1}.
  \end{equation*}
As in our setting there is a proportion $x$ of positive items with load smaller than $x$, we obtain
  \[
    \P ( V_\ell \le x| X_{1,1} = x) = (1 - p + px)^{n-1} = (1 - p(1-x))^{n-1} = g_{n,p}(x),
  \]
  which leads to the following upper bound for the false negative rate of an item with load $x$,
  \begin{equation*}
    FN(x) \leq  L\left( 1 - g_{n,p}(x) \right)^{L-1}.
  \end{equation*}

  We can similarly compute the false positive rate of the algorithm by computing the probability that conditionally on $(1,1)$ being non-contaminated, this item is determined to be contaminated. This would happen if and only if $(1,1)$ is only part of contaminated pools, and that the two pools with the lowest measured load have the same value, that we write $x$. Noticing that false positive results never occur in the infinite precision setting $K = \infty$, we assume here that $K < \infty$.
  Using again that we have a multipool and that the measure of each test is independent conditionally on the value of $X_{1,1}$ we obtain
  \begin{align*}
    FP(x) &=\P(I_{1,1}\ge 2,\ V_{1,1}=x | X_{1,1} = 0) \\
    &=\P(\exists \ell_1,\ell_2\in \mathcal{P}_{1,1}, V_{\ell_1}=V_{\ell_2}=x, \  \forall \ell \in \mathcal{P}_{1,1}\backslash \{\ell_1,\ell_2\}, V_\ell \ge x)\\
    &\le \sum_{\{\ell_1,\ell_2\} :\ell_1\neq \ell_2} \P(V_{\ell_1}= x|X_{1,1}=0) \P(V_{\ell_2}= x|X_{1,1}=0) \prod_{\ell' \in \mathcal{P}_{1,1}\backslash \{\ell_1,\ell_2\}} \P (V_{\ell'} \geq x|X_{1,1}=0)\\
    &=\frac{L(L-1)}{2}\P(V_{\ell_0}= x|X_{1,1}=0)^2\P (V_{\ell_0} \geq x|X_{1,1}=0)^{L-2},
  \end{align*}
  with $\ell_0$ a fixed element of $\mathcal{P}_{1,1}$.
  Let \(F(x)\) be the distribution function of \(V_{\ell_0}\) conditionally on $X_{1,1}=0$. Then, we can use the upper bound
  \begin{equation*}
    \P(V_{\ell_0} = x|X_{1,1}=0)=F(x)-F(x-1/K)\le \frac{\sup_{[x-\frac{1}{K},x]}F'(u)}{K}F(x)\le \frac{np}{K}(1-p(1-x))^{n-1}
  \end{equation*}
  We finally get
  \begin{equation*}
    FP(x) \leq \frac{L(L-1)}{2} \left(\frac{np}{K}\right)^2 g_{n,p}(x)^2 \left( 1 - g_{n,p}(x) \right)^{L-2}. \qedhere
  \end{equation*}
\end{proof}

In the rest of the article, we compute the optimal efficiency under different constraints, based on the above constructed pools, in different situations. We first consider non-adaptive strategies for detection of defective items based on the measure of lines and columns only, then adding eventually item tests for items whose status cannot be deduced by the first step algorithm. We then aim at optimal testing efficiency, assuming that samples can be infinitely divided, and recover results consistent with Mézard et al \cite{MTT08}. In Section \ref{sec:simulations}, we compare our asymptotic estimates with simulated experiments, and obtain the false positive/false negative rates and efficiency that can be archived in real testing conditions.

\section{Asymptotics of the false discovery rates at \texorpdfstring{\(L\)}{\textit{L}} fixed}
In this section, we derive equivalent expressions for the upper bound of the false discovery rates when the values of \(K\) and \(L\) remain fixed. We consider two asymptotic cases, when \(np \to 0\) and when \(np \to \lambda>0\). It is implicitly assumed that \(n \to \infty\) and \(p\to 0\). In the second case, the calculations are based on the Poisson approximation of the number of defective items in a specific pool and are consequently more accurate than the rates in the first case.

\paragraph{Case \(np \to 0\).}
In this case, \(g_{n,p}(x) \) can be lower bounded by \(g_{n,p}(0)\) since it is a increasing function and \(g_{n,p}(0)\sim e^{-np}\). Then
\begin{equation*}
  FNR \le Lp\E (1-g_{n,p}(U))^{L-1} \le Lp (1-g_{n,p}(0))^{L-1} \sim Lp(np)^{L-1}.
\end{equation*}
The false discovery rate is then upper bounded by
\begin{equation*}
  FPR \le K\frac{L(L-1)}{2} \left(\frac{np}{K}\right)^2 \left( 1 - g_{n,p}(0) \right)^{L-2}\sim\frac{L(L-1)}{2K}(np)^L
\end{equation*}
In particular, we see that in this regime (as \(K\) and \(L\) remain fixed), the false negative rate is small with respect to the false positive rate.

\paragraph{Case \(np \to \lambda >0\).}
In this situation, with $L$ being fixed, we immediately obtain that the number of defective items in each pool converges to a Poisson($\lambda$) random variable. We consider the asymptotic behaviour of the false positive and false negative rates obtained in this situation. We write
\[
  \bar{FNR}(\lambda,L;K) = \lim_{n \to \infty} p^{-1} FNR(n,L;\lambda/n,K) \quad \text{and} \quad \bar{FPR}(\lambda,L;K) = \lim_{n \to \infty} FPR(n,L;\lambda/n,K),
\]
as a function of $L$ and the precision $K$.

\begin{proposition}
  \label{prop:poisson_FNR_FPR}
  Under the condition \(np\to \lambda>0\), we have that
  \begin{equation*}
    \bar{FNR}(\lambda,L;\infty) = (1 + (L-1)e^{-\lambda})(1 - e^{-\lambda})^{L-1},
  \end{equation*}
  and, for all $K < \infty$,
  \begin{equation*}
    \bar{FNR}(\lambda,L;K) \leq L(1 -e^{-\lambda})^{L-1} \quad \text{and} \quad \bar{FPR}(\lambda,L;K)\le\frac{L(L-1)}{2K} (1 - e^{-\lambda})^{L-2}
  \end{equation*}
\end{proposition}

\begin{figure}[ht]
\centering
\includegraphics*[width=8cm]{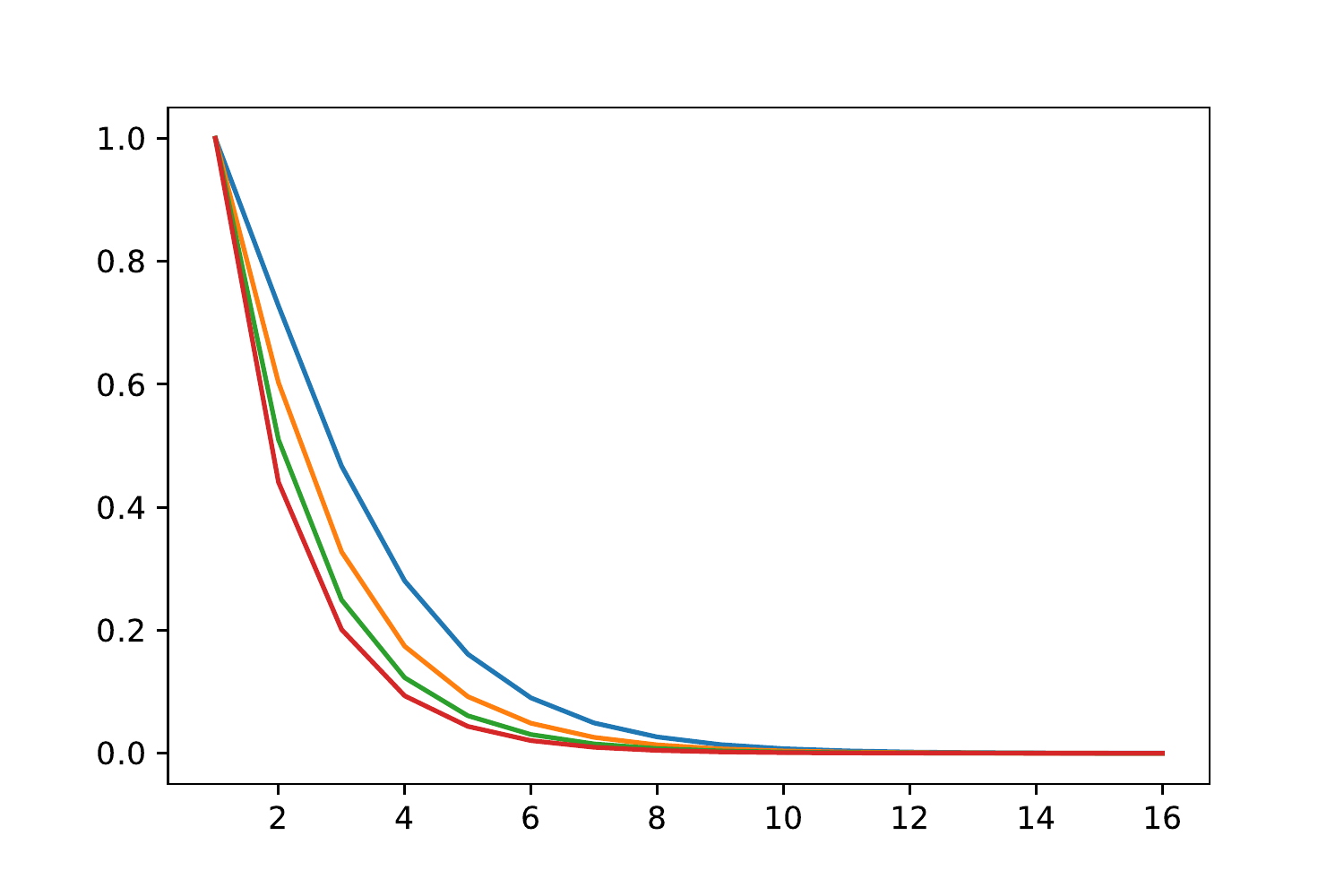}
\includegraphics*[width=8cm]{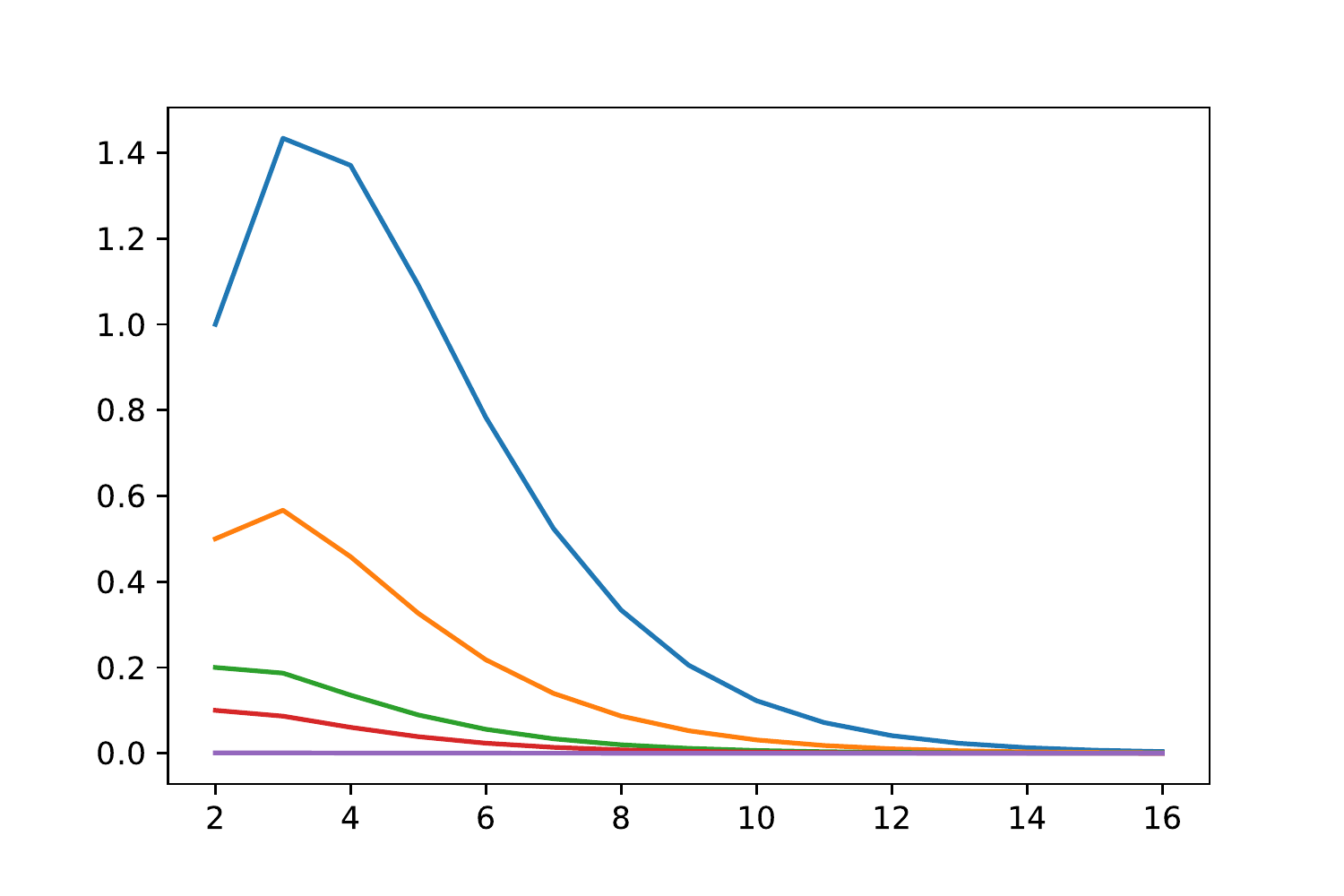}
\caption{False negative and false positive rates as a function of $L$ for $\lambda = \log 2$ and $K = 1$ (in blue), $K=2$ (orange), $K=5$ (green), $K=10$ (red) and $K = \infty$ (purple).}
\label{fig:falseNeg}
\end{figure}

\begin{proof}
  We first compute the false negative rate. From the properties of Poisson processes, we note that the number of items in a pool with load between $x$ and $y$ is distributed as a Poisson random variable with parameter $\lambda (y - x)$, independently of the number of items in this pool with load smaller than $x$ or larger than $y$. In particular, for any given test $\ell$, for any $y \geq x$ we have
  \[
    \P(V_\ell > y | X_{1,1}= x) = (1 - e^{-\lambda(1-y)}).
  \]

  Using this fact, with the same computations as in Proposition~\ref{prop:exact_upper_bounds}, we can compute the false negative rate of an item of load $x$ as
  \begin{align*}
    \bar{FN}(\lambda,L;K) &= L \sum_{y \geq x} \P(V_{\ell_0} = y | X_{1,1} = x) \P(V_{\ell_0} > y|X_{1,1} = x)^{L-1} \\
    &= L e^{-\lambda (1 - x)} (1 - e^{-\lambda(1-x)})^{L-1} + \sum_{y=x+1}^K e^{-\lambda(1-y)}(1 - e^{-\lambda/K}) (1 - e^{-\lambda(1-y)})^{L-1}\\
    &\leq L(1 - e^{-\lambda(1-x)})^{L-1}.
  \end{align*}
  In the case $K = \infty$ the computations can be made explicit as in that case
  \begin{align*}
    \bar{FN}(\lambda,L;\infty) &= L e^{-\lambda (1 - x)} (1 - e^{-\lambda(1-x)})^{L-1} + L \int_x^1 \lambda e^{-\lambda(1-y)} (1 - e^{-\lambda (1-y)})^{L-1} \dd y\\
    &= L e^{-\lambda (1 - x)} (1 - e^{-\lambda(1-x)})^{L-1} + (1 - e^{-\lambda(1-x)})^L = (1 + (L-1)e^{-\lambda(1-x)})(1 - e^{-\lambda(1-x)})^{L-1}.
  \end{align*}
  In particular, we also obtain $\bar{FN}(\lambda,L;\infty) \leq L(1 - e^{-\lambda(1-x)})^{L-1}$.

  Similarly, we can compute the false positive rate of a negative item in this regime. An item is falsely identified as a false positive if it does not belong to any pool with a negative item, and that if the smallest non-null value observed among the pools is attained at least twice.
  Using similar bounds as in the previous sections, we obtain
  \[
    \bar{FPR}(L;K) \leq \frac{L(L-1)}{2K^2} \sum_{k=1}^K (1-e^{-\lambda k/K})^{L-2}.
  \]
  Bounding the above quantity by $\frac{L(L-1)}{2K} (1 - e^{-\lambda})^{L-2}$, we obtain the result.
\end{proof}

From this formula, the false positive rate can be written explicitly as the probability that among $L$ independent copies with the above distribution, the first and second running minimum are equal. We plot once again this function in $L$ for different values of $K$.

\section{Optimizing one-step testing with  \texorpdfstring{\(L\)}{\textit{L}} fixed}
\label{sec:petitL}

In this section, we look to optimize Algorithm~\ref{algo:main} while assuming that the number $L$ of tests that can be made on each item is finite. This regime is relevant in particular if a test destroys or damages a sample of the item, so that limiting the number of tests made on each item becomes relevant. In the rest of the section $L$ is a fixed constant, and $n$ is a number with no prime factor smaller than $L-1$. In that situation Lemma~\ref{lem:nombre} holds and we are working with multipools, so that the formulas \eqref{eqn:genericFalseNegative} and \eqref{eqn:genericFalsePositive} both hold.

Recall that the efficiency of the algorithm is $E = \frac{L}{n}$, therefore to improve the efficiency of the algorithm, one has to increase the size of the grid. However, augmenting the value of $n$ has the effect of increasing the false positive and false negative rates. Therefore, to find the optimal efficiency of Algorithm~\ref{algo:main}, we fix $\epsilon > 0$ and $\eta > 0$ as maximal values for the proportion of positive and negative items wrongfully labelled as negative and positive respectively, and we choose $n$ as large as possible such that
\[
  FNR(n,L;p,K) \leq p\epsilon \quad \text{and} \quad FPR(n,L;p,K) \leq (1-p)\eta.
\]
As the average number of false negatives found by the algorithm is of the order $n^2p\epsilon$, we will also consider bounds in the regime when $\epsilon \to 0$ as $n \to \infty$.

\paragraph{A choice of \(n\) for $\epsilon$ and $\eta$ fixed.}
In this case, one has to choose \(n\) accordingly to have \(L(np)^{L-1} \le p \epsilon\) and \(\frac{L(L-1)}{2K}(np)^L \le \eta\). The largest \(n\) that satisfies the first condition is
\begin{equation*}
  n_1 = p^{-1}\left(\frac{\epsilon}{L}\right)^{1/(L-1)}
\end{equation*}
and the largest \(n\) that satisfies the second condition is
\begin{equation*}
  n_2 = p^{-1}\left(\frac{2K\eta}{L(L-1)}\right)^{1/L}.
\end{equation*}
We recommend to choose \(n\) as
\begin{equation}
  \label{eq:universal_choice}
  n\sim p^{-1} \min \left(\big(\frac{\epsilon}{L}\big)^{1/(L-1)} ; \big(\frac{2K\eta}{L^2}\big)^{1/L}\right)
\end{equation}
This choice of \(n\) gives a efficiency of the algorithm given by
\begin{equation}
E_{\epsilon,\eta}(p)=  p \max \left( \big(\frac{L^L}{\epsilon}\big)^{1/(L-1)} ; \big(\frac{L^{L+2}}{2K\eta}\big)^{1/L}\right).
\end{equation}
Note that these choices of values for \(n\) are driven by the results of Proposition~\ref{eqn:genericFalsePositive} and hence are quite conservative, so the efficiency obtained here is an upper bound of the true optimal efficiency of Algorithm~\ref{algo:main}.
Note that in a high precision setting (when \(K\) is large) the false positive are a minority inside the false discovery of the algorithm, and the efficiency will depend only on $p,L$ and $\epsilon$.

We observe that the number of tests to use per item to detect defective ones with fixed false negative/positive rate becomes proportional to $p$ the proportion of true negative. In other words, the total number of tests this algorithm need to detect defective items becomes ultimately proportional to the number of defective items when the number of non-defective items is large. This is a notable improvement on pool testing with $\{0,1\}$ response, in which the known optimal asymptotic efficiency is proportional to the product of the number of defective items and the log of the number of non-defective ones.

\paragraph{A choice of \(n\) for vanishing $\epsilon$ and $\eta$.}
Observe as well that in the settings we discuss, the expected number of false negatives in a given grid will grow as $\epsilon p n^2$ and the number of false positive as $\eta n^2$. It may be inconvenient to let the expected number of false discovery to grow as $n$ becomes large. To avoid a positive proportion of items on the grid being false negatives, one could instead consider a maximal false positive rate of $\epsilon = \alpha/(p n^{2})$ as $n \to \infty$ and $\eta = \beta/n^2$.

In this situation, with similar computations as above, we obtain optimal choices of
\[
  n_1 = \left(  \frac{\alpha}{Lp^L} \right)^{1/(L+1)} \quad \text{and} \quad n_2 = \left(  \frac{2 K \beta}{L(L-1) p^L} \right)^{1/(L+2)}
\]
and the associated efficiency to $L$ and $n = \min(n_1,n_2)$ is
\[
  E_{\alpha,\eta}(p) = \max \left(  p^{L/(L+1)}\big(\frac{L^{L+2}}{\alpha}\big)^{1/(L+1)} ; p^{L/(L+2)} \big(\frac{L(L-1)}{2K\beta}\big)^{1/(L+2)}\right) \quad
\]
This efficiency is much larger than $E_{\epsilon,\eta}(p)$, as expected from the lower tolerance to false negatives. It behaves as a power of $p$ as $p \to 0$. Remark that for $L = 2$, we have $E_{\alpha,\eta}(p) \sim C_{\alpha} p^{2/3}$ as $p \to 0$, so even with these settings, the algorithm becomes more efficient than Dorfman's method for $p$ small enough, to the cost of a fixed proportion $\eta$ of false positive. For $L = 3$, the algorithm becomes more efficient than Dorfman's algorithm even with a vanishingly small number of false positive.

\section{Optimal choice of \texorpdfstring{$L$}{L} as a function of \texorpdfstring{$p$}{p}}
\label{sec:Lmax}

In this section, we relax the assumption that $L$ has to be kept fixed, and aim at choosing an optimal couple $n,L$ so that the efficiency of the algorithm $E = L/n$ is as small as possible, while controlling the false positive and false negative rates. While it is not mentioned explicitly, it is assumed everywhere in this section that $L-2$ is smaller than the smallest prime factor of $n$, so that Lemma~\ref{lem:nombre} can be applied. Using the growth rate of primordial numbers and the fact that an optimal choice of $L$ will remain finite, there will always be a couple $(n,L)$ satisfying the assumption of Lemma~\ref{lem:nombre} close enough to the optimal theoretical choice, so this condition won't play a role in the asymptotic behaviour of the obtained efficiency.We first investigate the case when the precision of the loads \(K\) is infinite so that the choices of \(n\) and \(L\) are only driven by the reduction of the false negatives in the algorithm.

We take interest in the quantity
\begin{equation}
  E^*_\epsilon(p) := \min \left\{ \tfrac{L}{n}, n, L \in \N : FNR(n,L;p,\infty) \leq \epsilon p \right\}.
\end{equation}
In this new context, $L$ no longer fixed, and its choice might depend on $p$ and $\epsilon$.

We recall that by Proposition~\ref{prop:poisson_FNR_FPR} and the computations above, if $n p \to \lambda \geq 0$, we have
\[
  p^{-1} FNR(p,L;p,\infty) \lesssim L (np)^{L-1}.
\]
As a result we are lead to choose $n$ and $L$ such that $L (np)^{L-1} \approx \epsilon$, while minimizing the efficiency of the algorithm $E = \frac{L}{n} \approx p (\epsilon/L)^{1/(L-1)}$. We thus obtain that the efficiency of the algorithm is optimal when $L$ is taken to minimize
\[
  L \mapsto (\epsilon/L)^{1/(L-1)},
\]
and $n$ as $(\epsilon/L)^{1/(L-1)}/p$. Therefore, the optimum is attained by choosing $np \to \lambda > 0$.

To precise the computations in that situation, we use the formula given in Proposition~\ref{prop:poisson_FNR_FPR}. We observe that as $np \to \lambda$, we have
\[
  FNR(p,L;p,\infty) \to \bar{FNR}(\lambda,L;\infty) = (1 + (L-1)e^{-\lambda})(1 - e^{-\lambda})^{L-1}.
\]
Moreover, we have $E(p)/p = \frac{L}{np} \to L/\lambda$. As a result, we have
\[
  \lim_{p \to 0} \frac{E^*_\epsilon(p)}{p} =  \min \left\{ \tfrac{L}{\lambda}, L \in \N, \lambda > 0 : \bar{FNR}(\lambda,L;\infty) \leq \epsilon \right\}.
\]
As $\epsilon \to 0$, the optimal is attained for $\lambda,L$ such that $L \log (1 - e^{-\lambda}) = \log (\epsilon)(1 + o(1))$, yielding
\[
    \lim_{p \to 0} \frac{E^*_\epsilon(p)}{p} = \frac{-\log \epsilon(1 + o(1))}{-\lambda \log (1 - e^{-\lambda})}.
\]
This quantity is minimal for $\lambda = \log 2$.

As a consequence, as $p \to 0$, we recommand using $n \approx (\log 2)/p$, and as $\epsilon \to 0$ $L \approx -\frac{\log \epsilon}{\log 2}$, to obtain an optimal efficiency behaving as
\[
  E^*_\epsilon(p) \lesssim \frac{p(-\log \epsilon)}{(\log 2)^2} \quad \text{ as $p \to 0$ then $\epsilon \to 0$}.
\]
Note in particular that $n$ is chosen depending on the value of $p$, while $L$ is chosen as a function of $\epsilon$ in this asymptotic regime.

\begin{remark}
\label{rem:falseNegative2}
As noted in Remark~\ref{rem:falseNegative}, Algorithm~\ref{algo:main} could be adapted as a two-step algorithm in which every inconclusive item is tested again individually. Note that the upper bound we use for the false negative rate is exactly the rate of inconclusive elements (as we bound $FN(x)$ by $FN(0)$). In this two-step algorithm, the efficiency would therefore be $E^{(2)}(p) = E_\epsilon^*(p) + \epsilon \lesssim  - p \log ( \epsilon)(\log 2)^{-2} + \epsilon$, which is minimal for $\epsilon = (\log 2)^{-2} p$. This two-step algorithm would thus have an efficiency asymptotically bounded by $(\log 2)^{-2}p \log (p)$ as $p \to 0$.
\end{remark}

\begin{remark}
  We also observe that choosing $\epsilon = n^{-1-\alpha}$ for some $\alpha > 0$, the efficiency becomes
  \[
    E^*_\epsilon(p) \approx 2.08(1 + \alpha) p (-\log p ) \quad \text{ as }p \to 0.
  \]
  With this choice of $\epsilon$, the probability of observing one false negative in the grid decay to $0$ as $\epsilon n^2 p \approx n^{-\alpha}$. In that situation, we obtain an efficiency with a similar order of magnitude of the optimal results of \cite{MTT08}, with a simpler (non-random) construction of the algorithm. Actually, the efficiency obtained here attains the lower bound of the optimal efficiency predicted in \cite{MT11} for a two-steps binary pool testing algorithm. Therefore, the non-adaptive Algorithm~\ref{algo:main}, making use of the load value of items, archives the same efficiency as an optimal two-stage algorithm.
\end{remark}

\section{Comparison of the different algorithms}
\label{sec:simulations}

In this section, we illustrate the behavior of our proposed algorithm versus the two-steps Dorfman's algorithm and Mézard's optimal algorithm. We describe the choices of the parameters in the following.
\begin{itemize}
  \item The simulations took the set \(\{3,5,7,11,13,17,19,23,29,31,37,41,43,47\}\) of the odd prime numbers under 50 as the possible values of \(n\).
  \item The algorithm is allowed to perform tests in up to \(L_{\max}=14\) directions. In the case when \(n\le 14\), we obviously restrict the number of directions to \(L_{\max}=n-1\).
  \item The prevalence parameter \(p\) varies between 0.05 and 0.2 with a constant increment of \(0.05\).
  \item We made vary \(K\) inside the set \(\{2,5,10,30,200,500\}\) to illustrate its influence.
  \item For each choice of the parameters \((n,L,p,K)\) above, we run 200 copies Algorithm~\ref{algo:main}.
\end{itemize}
Consequently, for each choice of the set of parameters, we observe 200 copies of the output of the algorithm. Afterwards, the matrix of results is compared to the matrix of the true matrix containing the information of the true positive and negative items. Thanks to that, we compute the mean number (over the 200 copies) of false positive and false negative discovered by the algorithm. Thus, we end with an estimation of the number of false negative \(FN(n,L,p,K)\) and the number of false positive \(FP(n,L,p,K)\). The next step is to compute the optimal value of the efficiency $E$ as a function of \(p\). Then, for any couple \((p,K)\) fixed, we did the following concrete inclusions of the conditions of Section~\ref{sec:petitL}.
\begin{enumerate}
  \item We fixed \(\eta=0.01\) and we discarded all the pairs \(n,L\) such that \(FP(n,L,p,K)\ge \eta (1-p)n^2\).
  \item For each value of \(\epsilon\) we considered, we discarded the pairs such that \(FN(n,L,p,K)\ge \epsilon pn^2\).
  \item Then, from all the remaining values of the pairs \((n,L)\), we minimized the quotient \(E(p)=L/n\).
\end{enumerate}
This value \(E(p)\) as a function of \(p\) (for different values of \(K\)) is the one that we drew in the following illustrations.

\subsection{Comparing the Efficiency with well known algorithms}
\begin{figure}[ht]
  \includegraphics[scale=0.45]{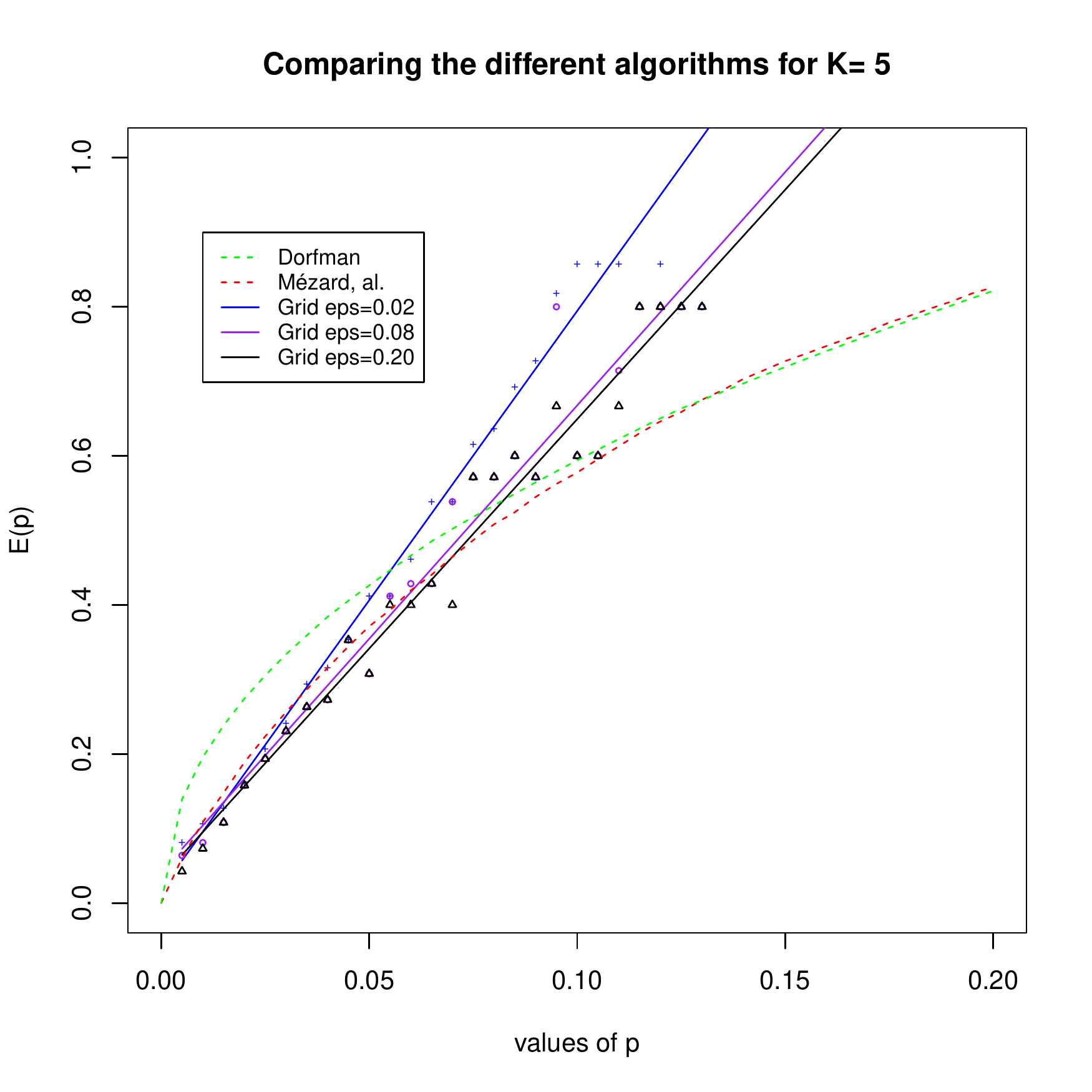}
  \includegraphics[scale=0.45]{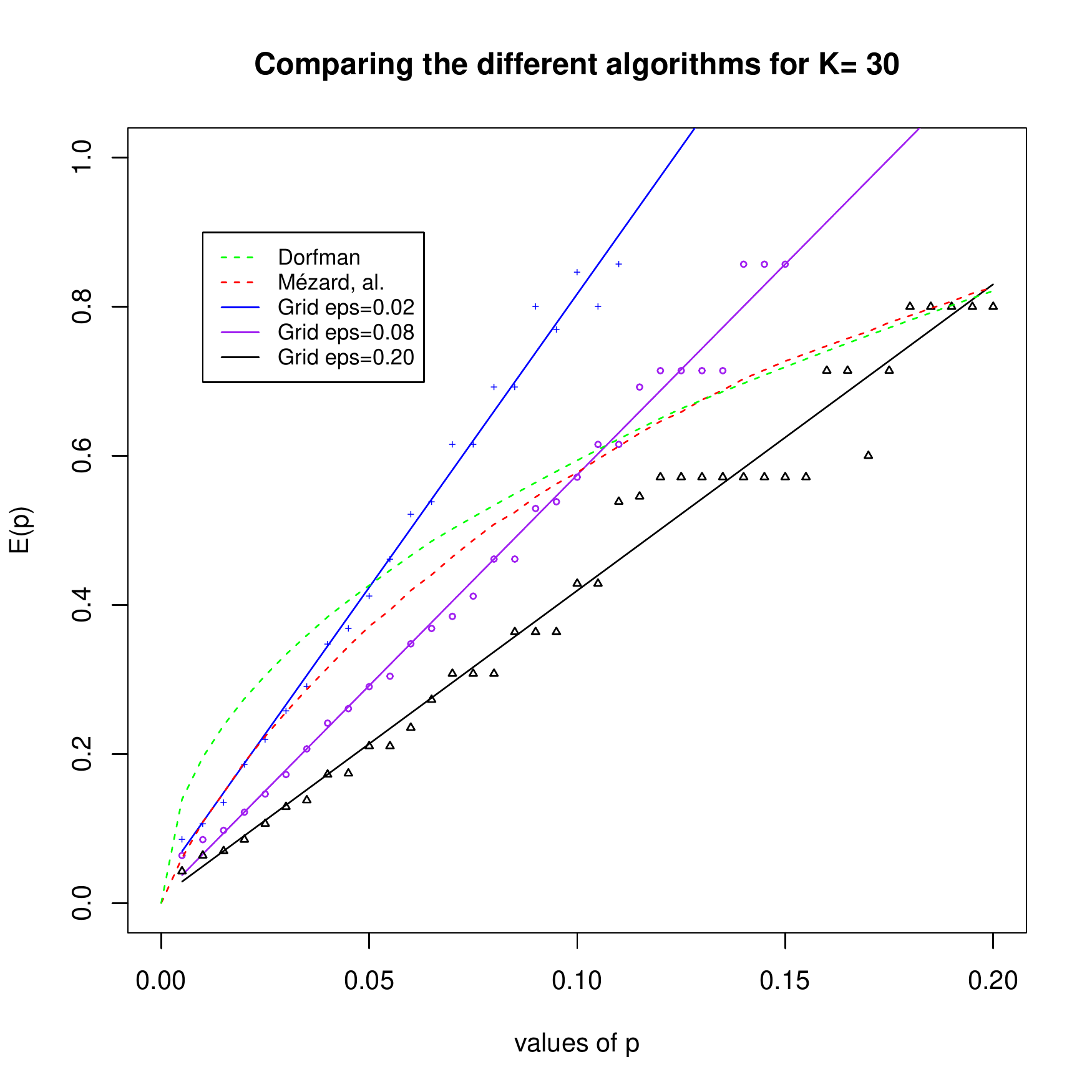}
  \caption{The effect of the value of \(K\) on the slope of \(E(p)\) for different values of \(\epsilon\)}
  \label{fig:Compar_E_Dorf_Mez}
\end{figure}
Figure~\ref{fig:Compar_E_Dorf_Mez} shows for the two different values \(K=5\) and \(K=30\), the behavior of our Efficiency curve versus Dorfman theoretical efficiency and a simulated Mézard, al. efficiency. We drew the resulting points of \(E(p)\) in three different colors (blue,purple,black) that correspond to the choices of \(\epsilon\) given by \((0.02,0.08,0.2)\). For each of these ensembles of points, we also drew a simple regression line. It has to be seen that the dependence of \(E\) on \(p\) is clearly linear and that the slope of the line is dependent on the choice of the parameter \(\epsilon\), as expected. It is also interesting to see that the effect of \(\epsilon\) is less clear when \(K\) is small since the number of false positives is higher and then is more limitent than when \(K\) is large.

\subsection{Showing the choices of \texorpdfstring{\(L\)}{L} and \texorpdfstring{\(n\)}{n}}
The next three plots (in Figure~\ref{fig:Ldependence}) show the choices of the parameter \(L\) during the optimization of \(E\) for fixed values of \(p\) and \(K\). As before we let \(\epsilon\) vary in between the different plots. Besides been a little unstable in the choice of \(L\) along \(p\), we observe that the optimal \(L\) remains bounded (hence is fairly independent from the choice of \(p\)) and does change with a change of \(\epsilon\) as suggested by the calculations in Section~\ref{sec:Lmax}.
\begin{figure}[ht]
  \includegraphics[scale=0.3]{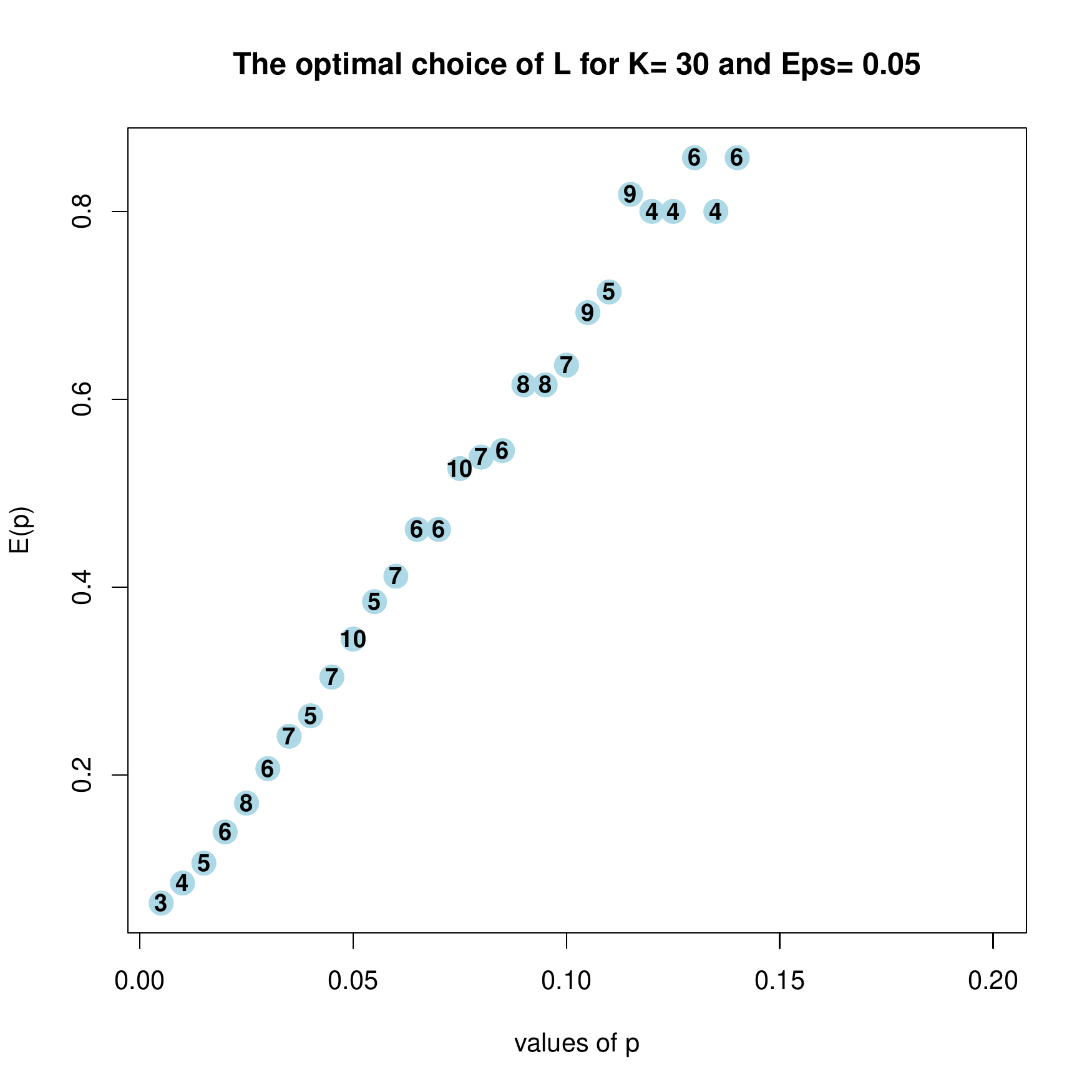}
  \includegraphics[scale=0.3]{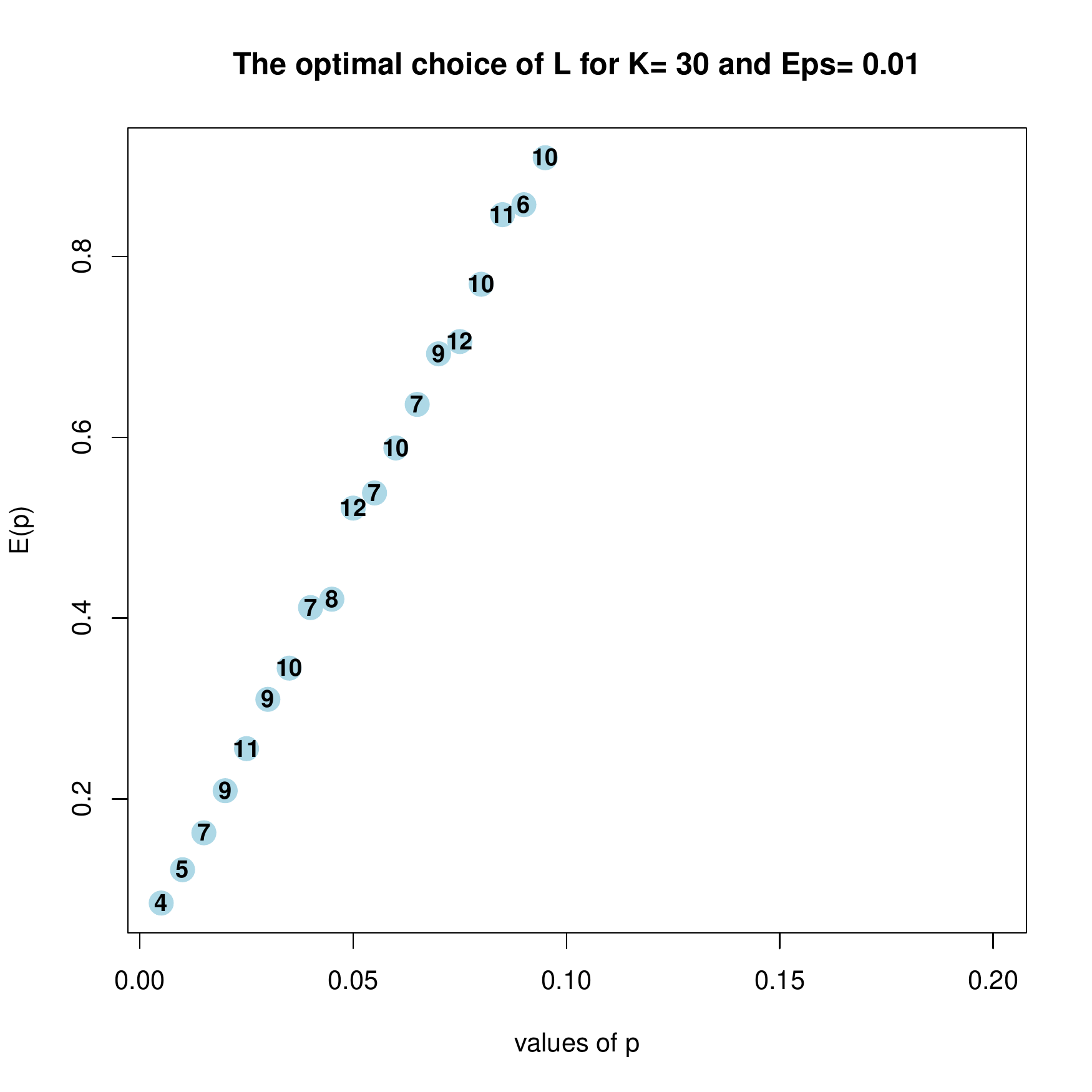}
  \includegraphics[scale=0.3]{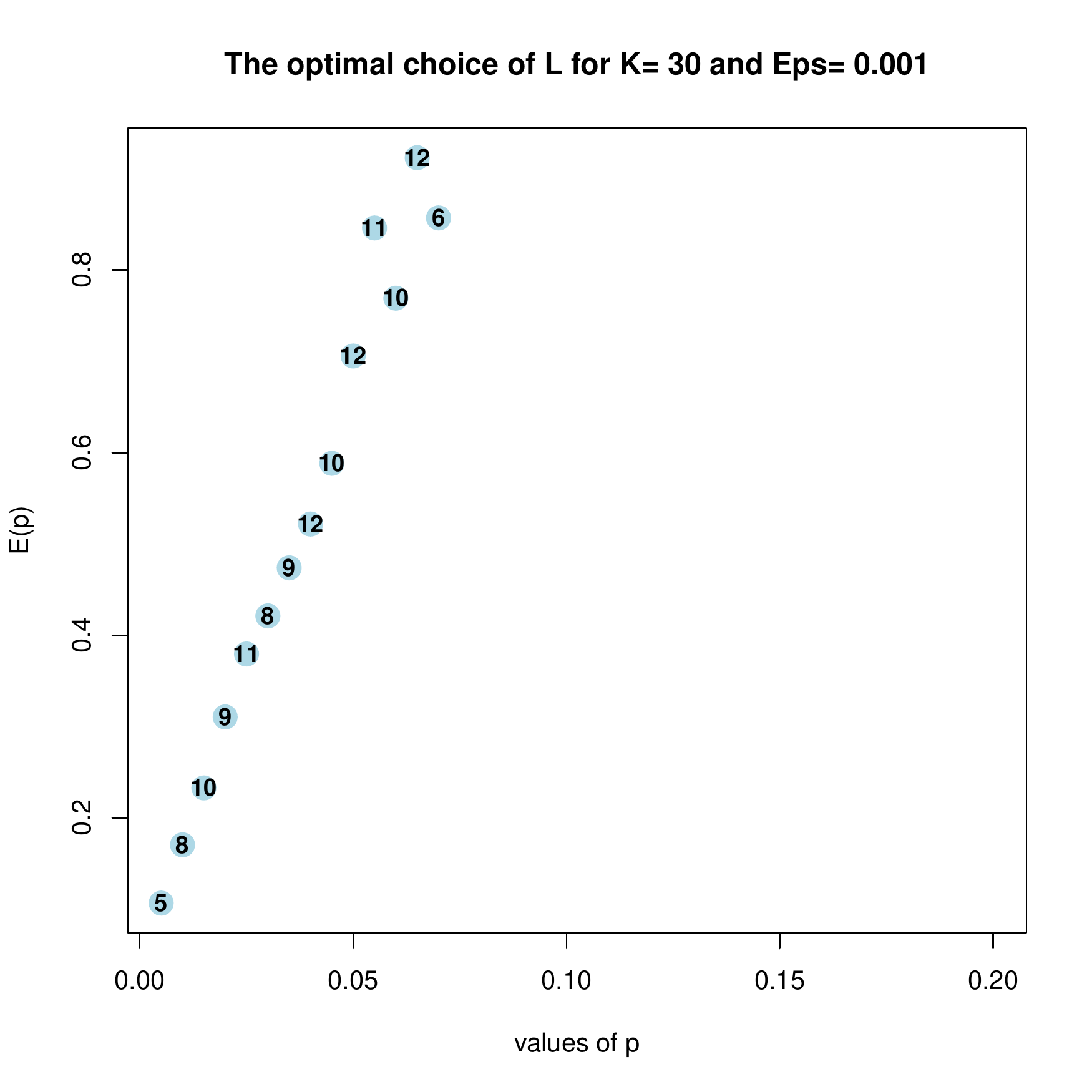}
  \caption{Efficiency with respect to \(p\) and the associated optimal choice of the parameter \(L\). The number displayed inside the blue bubbles correspond to the chosen value of \(L\) in the optimization.}
  \label{fig:Ldependence}
\end{figure}

The last three plots (in Figure~\ref{fig:ndependence}) are the analogs of the previous plots with the slight difference that the displayed numbers correspond to the chosen values of \(n\). In this case, we observe that, now, \(\epsilon\) has no more effect on the chosen values of \(n\). As expected, \(n\) depends on \(p\) in a decreasing manner and validate the calculation of Section~\ref{sec:Lmax}. Indeed, we showed that the best choices of \(n\) allow to keep the product \(np\) more or less constant which is the case in the simulations.
\begin{figure}[ht]
  \includegraphics[scale=0.3]{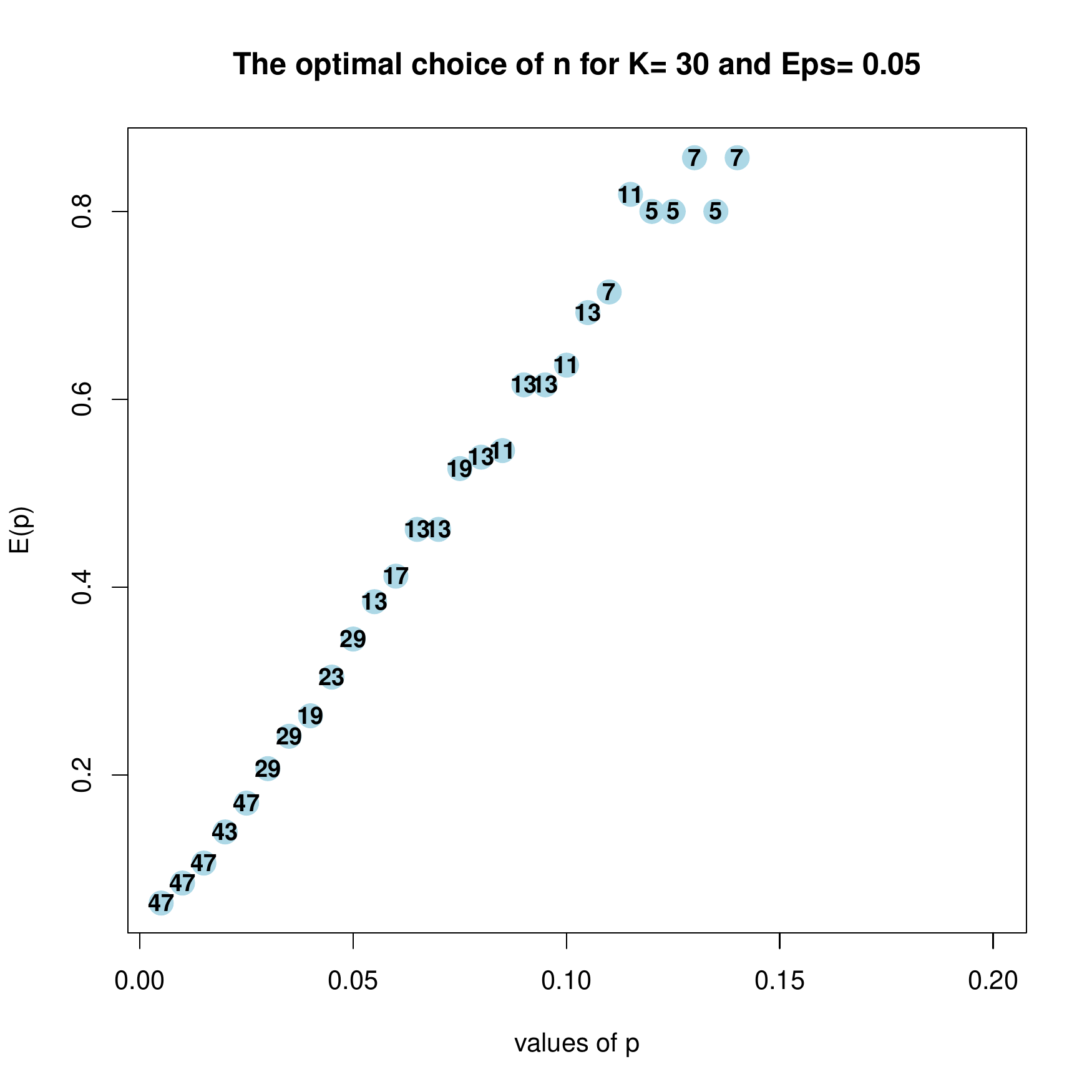}
  \includegraphics[scale=0.3]{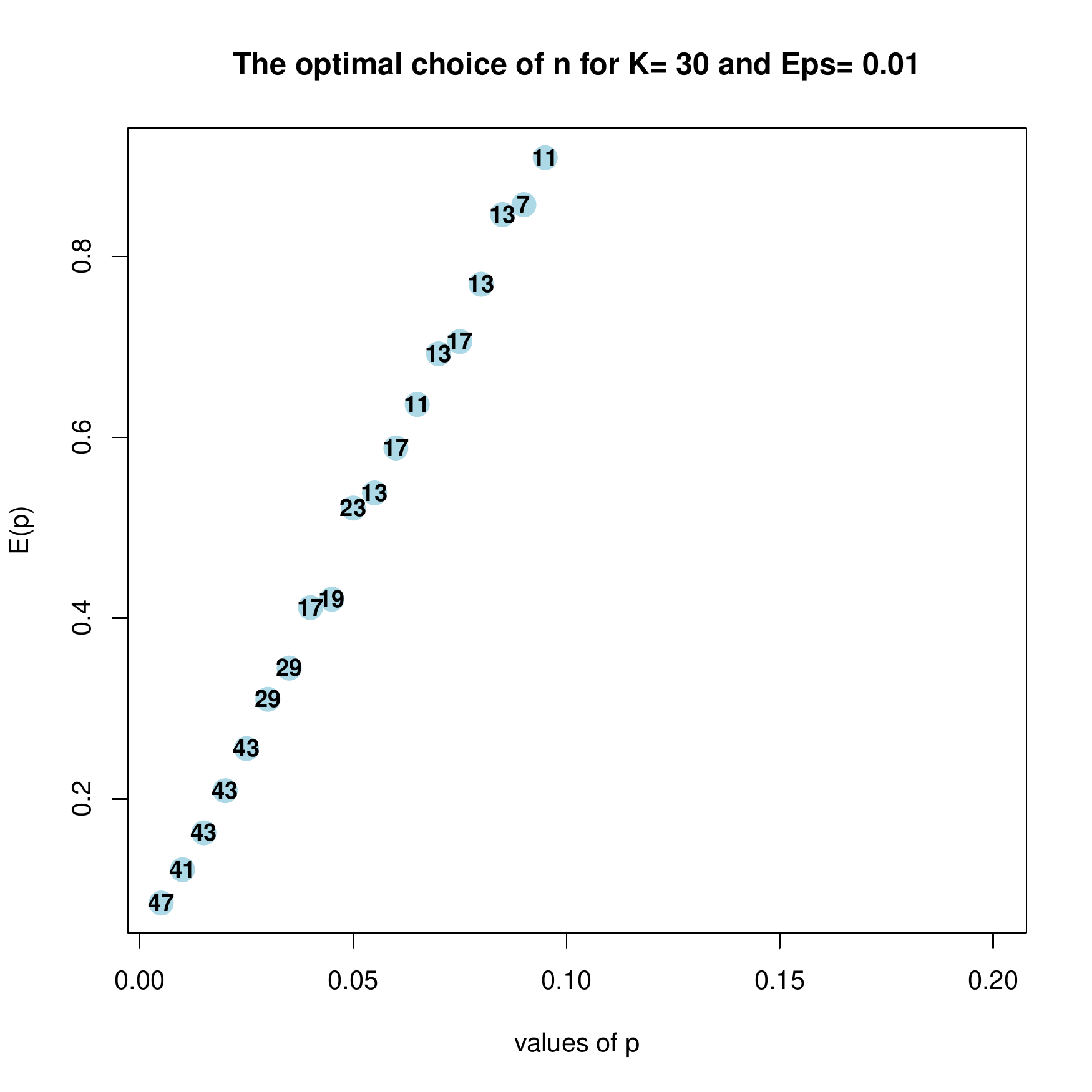}
  \includegraphics[scale=0.3]{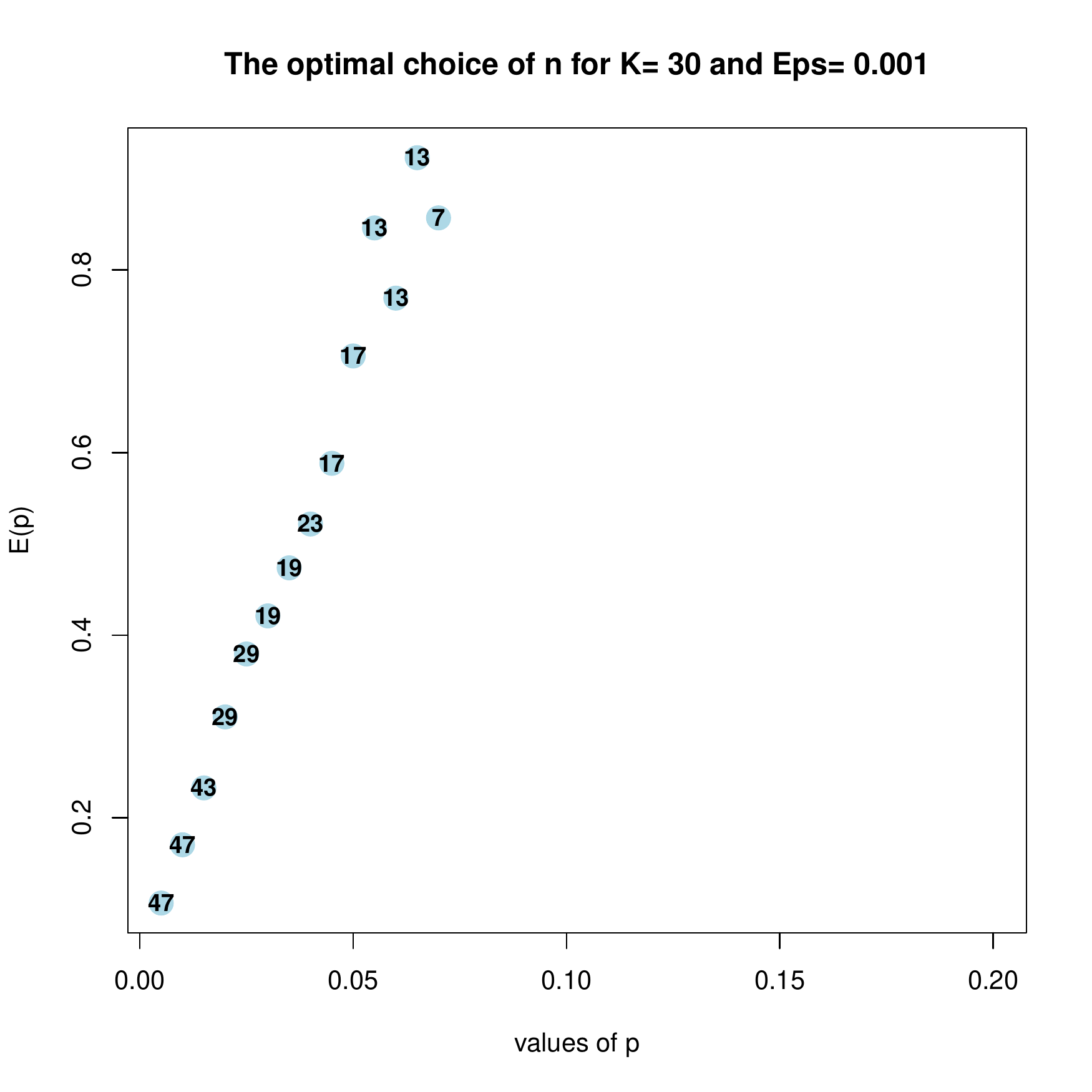}
  \caption{Efficiency with respect to \(p\) and the associated optimal choice of the parameter \(n\). The number displayed inside the blue bubbles correspond to the chosen value of \(n\) in the optimization.}
  \label{fig:ndependence}
\end{figure}

\section{Application to the COVID-19 pandemic and open questions}
\label{sec:covid}

The application of the present algorithm to PCR testing in the context of the COVID-19 pandemic requires some adaptation and presents a couple of challenges. It should first be noted that the simplifications we made in our modeling were quite important. We thus begin by discussing in more details the discrepancies between the real-world problem and our idealized model.

\paragraph{Finite size of samples.}
In COVID-19 pool testing, the items that are tested are samples taken from subjects, via nasal swab, saliva sample or other method. If there seems to be usually enough matter to split the sample into several tests, it will not be possible to make an arbitrary large number of tests on each sample. Therefore, optimal computations made in Section~\ref{sec:petitL} might be somewhat more relevant. Additionally, it is worth noting that combining several samples have the effect of creating a composition with the average viral load rather than the maximal, although the fact that this viral load is spread over several orders of magnitudes negates partially this problem as discussed in the introduction.

\paragraph{Noisiness in the measure.}
We chose to represent a lack of accuracy of the measure by replacing the continuously distributed viral load by a discrete distribution, as if contaminated tests could be arranged into ``classes'' of similarly measured viral load. In practice, the quantity returned is a real number, which is a measure on which a Gaussian noise is expected to be applied. More precisely, several steps of the process have the effect of creating uncertainty on the measure, and the variance that has to be expected from a pooled test might be rather large. There is first the collection of the portions of samples used to make the group testing, whose volume and associated viral load may vary with regards to the attended equal contributions. Next, the reverse transcriptase step that converts RNA samples of the virus into DNA might introduce additional noise depending on its rate of conversion. Finally, the PCR measurement itself products a noisy value, partially corrected by the fact that classically, two different DNA sequences for the virus are measured separately. It would then require some adaptation of our algorithm to adapt the ``finite precision algorithm'' to noisy Gaussian measure of the viral load in each pool, although classical likelihood ratio estimates might be successfully used here.

\paragraph{Distribution of the viral load among contaminated.}
Concerning the distribution of the viral load of the samples, we made here the choice of uniform distribution, which is the most favorable for this type of algorithm. Although this is quite far from what is effectively observed \cite{Jones2020,Cabrera2020}, the viral load observed among large groups of people is usually successfully approached by a mixture of two to three Gaussian variables with standard deviations between 3 and 6, spanning over the interval $[20,40]$ (c.f. \cite[Appendix B]{BMR}). The viral loads may be considered sufficiently spread over the interval so that the algorithm discussed above might still be relevant. The nice and explicit calculations on the choices of the parameters would need to be adapted. They might also need to be tuned from day to day, depending of the expected prevalence of samples on a given day, which might vary over time.

\paragraph{Precision limits of the PCR.}
As it was first noted in \cite{Furon2018}, theoretical aspects of pool testing usually assume that the quality of the test does not depend on the size of the pool. However, this is rarely the case in real-world applications, and it is indeed not the case for the present application. In particular, we show in \cite{BMR} that pooling has an impact on the measurement of samples with small viral loads, due to a dilution effect. In our toy model, this could be taken into account by specifying that in pools of size $n$, items with load smaller than $c_n$ are treated as non-defective items, for some increasing function $c_n$. This has the effect of decreasing the value of the optimal choice for $n$, in order to detect enough contaminated individuals with small viral load. However, the computations in this case being very dependent on the function $c_n$, we choose not to include it in the present work.

\paragraph{Random outcome of a pooled test.}
Finally, we assumed that each test of the pools is performed with no other error than the one inherent to the PCR itself. It is probably an oversimplification in this case, as pool testing implies important manipulations of the samples, with possible additional errors involved. For example, forgetting to collect one individual in a pool, contaminating a pool with a sample that should not belong to it, etc. Those human errors would create noise on the measures of the pools and so would deteriorate the information given to our algorithm. Therefore, it would need to be adapted to this situation, in order not to characterize as negative a sample measured at high values in all but one pool, for example. There is also the issue of systemic bias in PCR, as most machines only allow the measure of the relative viral load rather than an absolute value. As our algorithm only consider relative loads, this is not generally an issue for our method, when the algorithm is performed on a single machine.

\paragraph{Potential extensions.}
The algorithm presented here has the advantage of being simple to implement and easy to solve, even by hand. However, more precise algorithms might be employed with the help of automation for the creation of samples and measure of results. It would therefore be interesting to create more precise algorithms for PCR-type pool testing.
A relevant generalization could be to collect and use additional information on the subjects. We can imagine that, throughout interviews, some individuals might be identified as being more likely to be contaminated, while others could be simply routinely tested. It is probably more efficient to tests the former in smaller pools and the latter in larger ones.

An other project of interest might be the deconvolution of pools created by Dorfman's algorithm. More precisely, in the algorithm, instead of testing individually every member of a group detected as contaminated, it might be interesting to test several samples from different positive pools in a two-stage deconvolution that might represent a further economy of tests on Dorfman's algorithm. Choosing the right pools to pair together, as well as the number of positive pools to be de-convoluted at the same time might be an interesting expansion on the current work.

\subsubsection*{Acknowledgements.}
We wish to thank members of MODCOV-19 platform of the CNRS for support, in particular Françoise Praz and Florence Débarre who gave us numerous helpful comments, in particular on the biological aspects of PCR. We also thank the members of the Groupool initiative for helpful discussions at the earlier stages of this project on group testing.


\newcommand{\etalchar}[1]{$^{#1}$}

\end{document}